\documentclass[aps,prl,reprint,amsmath,amssymb,superscriptaddress,nobibnotes,showpacs,showkeys]{revtex4-1}

\pdfoutput=1

\usepackage{dcolumn}
\usepackage{bm}
\usepackage{amsmath}
\usepackage{amssymb,amsthm}
\usepackage{bbm}
\usepackage{epsfig,color}

\usepackage[sans]{dsfont}
\usepackage{comment}

\PassOptionsToPackage{hyphens}{url}
\usepackage{hyperref}
\usepackage{url}

\usepackage{units}
\usepackage{color}
\usepackage{ifthen}
\usepackage{graphicx}
\graphicspath{ {images/} }

\usepackage{mathrsfs}

\usepackage[utf8]{inputenc}
\usepackage{tikz}
\usetikzlibrary{automata,positioning}
\usetikzlibrary{fit,shapes.geometric}
\usetikzlibrary{decorations.pathmorphing}
\usetikzlibrary{arrows,matrix}
\usepackage{tikz-3dplot}
\usetikzlibrary{patterns}

\usepackage{mathtools}
\DeclarePairedDelimiter{\abs}{\lvert}{\rvert}

\definecolor{nblue}{rgb}{0.2,0.2,0.7}
\definecolor{ngreen}{rgb}{0.2,0.6,0.2}
\definecolor{nred}{rgb}{0.7,0.2,0.2}
\definecolor{nblack}{rgb}{0,0,0}

\DeclareMathOperator{\rank}{rank}

\newcommand{\ket}[1]{|#1\rangle}
\newcommand{\ketbra}[2]{\ket{#1}\!\bra{#2}}
\newcommand{\bra}[1]{\langle#1|}

\newcommand{\tr}{\text{tr}}

\def\A{\mathcal{I}}

\def\I{\mathcal{I}}

\newcommand\pinv[1]{#1^{+}}

\newcommand{\N}{\mathcal{N}}
\newcommand{\E}{\mathcal{E}}
\newcommand{\Pol}{\mathcal{P}}
\newcommand{\x}{\textrm{x}}
\newcommand{\D}{\mathcal{D}}

\newcommand{\pguess}{P_\textrm{g}}

\newtheorem*{theorem*}{Theorem}

  \theoremstyle{definition}
  \newtheorem{defn}{\protect\definitionname}
  \theoremstyle{plain}
  
\theoremstyle{plain}
  \newtheorem{prop}{\protect\propositionname}
\theoremstyle{plain}

  \theoremstyle{plain}
  
  \theoremstyle{plain}

  \providecommand{\conjecturename}{Conjecture}
  \providecommand{\definitionname}{Definition}
  \providecommand{\lemmaname}{Lemma}
\providecommand{\corollaryname}{Corollary}
\providecommand{\theoremname}{Theorem}
\providecommand{\propositionname}{Proposition}

\def\x{\mathrm{x}}
\def\y{\mathrm{y}}

\def\ido{\mathrm{id}}

\def\E{\mathcal{E}}
\def\N{\mathcal{N}}

\def\OMP{\mathcal{OMP}}
\def\I{\mathcal{I}}

\def\id{\mathbb{I}}

\def\D{\mathcal{D}}

\def\tr{\mbox{tr}}

\def\bea{\begin{eqnarray}}
\def\eea{\end{eqnarray}}

\definecolor{CircleBlue}{RGB}{0, 163, 232}
\definecolor{CircleYellow}{RGB}{254, 242, 0}
\definecolor{CircleOrange}{RGB}{255, 127, 38}

\begin{document}

\title{Optimal measurement preserving qubit channels}

\author{Spiros Kechrimparis}
\email{skechrimparis@gmail.com}
\author{Joonwoo Bae}
\email{joonwoo.bae@kaist.ac.kr}

\affiliation{School of Electrical Engineering, Korea Advanced Institute of Science and Technology (KAIST), 291 Daehak-ro Yuseong-gu, Daejeon 34141 Republic of Korea,}
	
	\begin{abstract}
	
	We consider the problem of discriminating qubit states that are sent over a quantum channel and derive a necessary and sufficient condition for an optimal measurement 
	to be preserved by the channel. 
	We apply the result to the characterization of \emph{optimal measurement preserving} (OMP) channels for a given qubit ensemble, e.g.,  a set of two states or a set of multiple qubit states with equal \emph{a priori} probabilities.
	 Conversely, we also characterize qubit ensembles for which a given channel is OMP, such as unitary and depolarization channels. Finally, we show how the sets of OMP channels for a given ensemble can be constructed.
	

	\end{abstract}

	\maketitle


\section{Introduction}

In quantum information processing tasks, the problem of state discrimination often lies at their core. State discrimination can be simply described as a communication scenario between two parties: 
quantum states from a known ensemble are transmitted from a sender to a receiver through a quantum channel. The aim of the receiver is to perform measurements on the states in order to identify them. 
In \emph{minimum error discrimination} the success of the task is evaluated by the \emph{guessing probability}, the probability of guessing the received state correctly on average \cite{helstrom1969,bae2015,barnett2009}. It is known that, in general, this task cannot be performed without error and, due to this fact, state discrimination finds application in quantum information, quantum communication and quantum foundations\cite{bae2015, bae2013-2}. 


In order to achieve the guessing probability, a measurement apparatus has to be optimized to identify states  from a given ensemble. If the states in the ensemble are altered, so is, in general, the optimal measurement. 
From a practical point of view this poses a problem since in realistic scenarios unknown sources of noise inevitably exist. 
The noisy states would, then, need to be identified by means of state or channel tomography, both of which are experimentally costly.
Thus, it is important to characterize 
OMP channels, i.e. channels that preserve an optimal measurement for state discrimination.
 

In  \cite{kechrimparis2019} the notion of OMP channels has been introduced and a condition has been proven for the characterization of channels that do not change an optimal measurement for state discrimination of a given ensemble. Moreover, the depolarizing channel has been shown to be OMP for (i) equiprobable ensembles, and (ii) two-state ensembles \cite{kechrimparis2020,*kechrimparis2020a}.  Based on this observation a protocol has been proposed mapping an unknown channel into an OMP one by using LOCC only, saving the cost of performing channel or state tomography. Interestingly, this pre- and post-processing also increases the guessing probability for certain ensembles and channels.
However, the condition for the characterization of OMP maps in \cite{kechrimparis2019} is only sufficient and of limited use in the case of ensembles with unequal \emph{a priori} probabilities.

 The aim of this work is to completely characterize OMP channels for ensembles consisting of qubit states. To that end, we derive a necessary and sufficient condition for the preservation of an optimal measurement for the discrimination of qubit states. Beyond qubit ensembles, our condition becomes only sufficient but can be applied to general ensembles of unequal probabilities, in contrast with previous results.

The paper is structured as follows. We first review known results from the theory of state discrimination. We then present the main result, a necessary and sufficient condition for the preservation of an optimal measurement for the discrimination of states from a qubit ensemble, sent through a quantum channel. We discuss certain special cases, such as ensembles of equal \emph{a priori} probabilities and ensembles of two states only, and subsequently examine the OMP properties of depolarizing and unitary channels. Finally, we show how one can construct the general form of qubit channels that are OMP for a given ensemble and optimal measurement, and provide a number of examples.

\section{Preliminaries}
A \emph{quantum state} $\rho$ is described by a Hermitian, positive semidefinite operator of trace one acting on a Hilbert space $\mathcal{H}$.
An \emph{ensemble} $S$ of states is a collection of known states $\rho_\x$ that appear with \emph{a priori} probabilities $q_\x$. We denote such an ensemble of $n$ states with $S=\{q_\x,\rho_\x\}_{x=1}^n$. A measurement is represented by a positive operator-valued measure (POVM), a collection of positive semi-definite operators that sum to the identity operator, i.e., $M=\{M_\x\}_{\x=1}^n$, with $M_\x \geq 0$ and $\sum_\x M_\x =\id$. 

In minimum error discrimination one can assume, without loss of generality, that the optimal POVM has the same number of elements $n$ as the number of states in the ensemble, since some of these elements can be the null operator. However, if some of the operators in the POVM are the null operator, we keep track of the non-zero elements by considering the index set $\A$ that collects the indices of the operators that are strictly non-zero. In that case, we say that the states with indices not in $\A$ are \emph{not identified} by the measurement; in other words, these states are never detected by the measuring apparatus.

The goal of minimum error state discrimination is to find the  measurement strategy that minimizes the error in guessing the states of a known ensemble correctly, or equivalently, to find the strategy that maximizes the probability of guessing correctly on average.
The problem of minimum error state discrimination can be mathematically stated in the following way: given an ensemble of states $S=\{q_\x,\rho_\x\}_{\x=1}^{n}$, find the POVM $\{M_\x\}_{\x=1}^n$ that maximizes the guessing probability
\begin{equation}
\pguess =  \max_{\{M_\x\}}\, \, \sum_{\x=1}^{n}q_\x \tr[M_\x \rho_\x]\,,
\end{equation}
subject to the constraints $M_\x\geq 0$ and $\sum_\x M_\x =\id$.
Closed form solutions exist only in limited cases, such as a pair of qubit states, or ensembles of states with symmetries \cite{helstrom1969,bae2013,bae2015,barnett2009}.

A different formulation of the problem that has a geometric flavor is one based on the \emph{linear complementarity problem} (LCP) \cite{bae2013-2,bae2013}. In this approach, one is looking for a POVM $M=\{M_\x\}_{\x=1}^n$, a \emph{symmetry operator} $K$, \emph{complementary states} $\sigma_\x$ and non-negative numbers $r_\x$ that obey the so called \emph{Karush-Kuhn-Tucker(KKT)} conditions:
\begin{align}
& \, \, K=q_\x \rho_\x + r_\x \sigma_\x \,, \notag \\
& r_\x \tr[M_\x \sigma_\x] = 0 \quad \forall \x. \label{eq: geometric KKT}
\end{align}
Once these are identified, the guessing probability is obtained through
\begin{equation}
\pguess = \tr K =q_\x + r_\x \,, \quad \forall \x \,. \label{eq: pguess KKT}
\end{equation}
It is worth noting that the first of the KKT conditions can be re-written in the following form
\begin{equation}
q_\x \rho_\x - q_\y \rho_\y = r_\y \sigma_\y -r_\x\sigma_\x \,, \label{eq: 1st KKT geometric form}
\end{equation}
which unveils the geometric flavor of this approach. Indeed, let us first define the original polytope for the set of points constructed by multiplying each state with the respective \emph{a priori} probability $q_\x$, i.e.  $\Pol=\{q_\x \rho_\x\}_{\x=1}^n$, residing in the space of $2\times 2$ Hermitian matrices. Similarly we define the polytope of the complementary states, $\Pol_c=\{r_\x \sigma_\x\}_{\x=1}^n$. Then the meaning of the KKT condition in Eq.\@ \eqref{eq: 1st KKT geometric form} becomes clear: the original and complementary polytopes are congruent. We note that the complementary states $\sigma_\x$, the non-negative parameters $r_\x$, as well as the symmetry operator $K$ are unique. Having obtained $r_\x \,, \sigma_\x$ and using the second of the KKT conditions, Eq.\@ \eqref{eq: parameters to pguess}, one can obtain optimal measurements that solve the discrimination problem. However, note that an optimal measurement is not unique in general.

When it comes to qubit state discrimination, the symmetry operator $K$ determines the possible measurement strategies. Specifically, the following cases exist \cite{weir2017}:
\begin{enumerate}
	\item $K-q_j \rho_j=0$ for some $j$. Then, it follows that $K=q_j \rho_j$ and this can occur if $q_j \rho_j -q_k \rho_k \geq 0, \quad \forall k$. In such case the optimal strategy consists of not performing a measurement and always guessing state $\rho_j$, i.e. the POVM elements are $M_\x = \id \delta _{j \x}$.
	\item $K-q_j \rho_j>0$. If this operator is positive definite, the POVM element $M_j$ is the null operator in every optimal measurement. Thus, the state $\rho_j$ is never identified by any optimal measurement strategy.
	\item $K-q_j \rho_j$ has a single zero eigenvalue. Then, $M_j$ is a weighted projector, $M_j= w_j \ketbra{\phi_j}{\phi_j}$, where $0\leq w_j\leq 1$ and $\ket{\phi_j}$ is the eigenstate corresponding to the zero eigenvalue.
\end{enumerate}
These three cases cover all possibilities. For the problem we are considering in this work the first two cases are trivial and can be excluded, without loss of generality. For the first case this is obvious, since if the optimal strategy consists of no measurement, then there is no measurement to be preserved: regardless of the effect of the channel, one can always guess according to the \emph{a priori} probabilities. The second case can also be ignored since if a state in the ensemble is never identified by a measurement, one can remove that state from the ensemble and re-define the \emph{a priori} probabilities. Specifically, if $S=\{q_\x,\rho_\x\}_{x=1}^n$ is the original ensemble and state $\rho_1$, say, is never identified by an optimal measurement, then one can consider the ensemble $S^\prime = \{q^\prime_\x,\rho_\x\}_{x=2}^n$ , where $q^\prime_\x = \nicefrac{q_\x}{r}$ and $r=\sum_{j=2}^n q_j$. The two ensembles are equivalent from a state discrimination point of view, since they have the same optimal measurements and complementary states. Moreover, the symmetry operators, guessing probabilities, as well as parameters $r_\x , r^\prime_x$ are related by a rescaling, that is, $K^\prime =\nicefrac{K}{r}$, $\pguess^\prime = \nicefrac{\pguess}{r}$, and $r^\prime_\x = \nicefrac{r_\x}{r}$. After these considerations, without loss of generality we can assume that every state in the ensemble will be identified by some optimal measurement. 

\section{Preservation Of Optimal Measurements\label{sec: main results}}
Let us start with a few definitions before we present the main result.
\begin{defn}
	We call a channel $\N$ \emph{optimal measurement preserving} (OMP) for an ensemble of states $S=\{q_\x ,\rho_\x\}_{\x=1}^{n}$, if an optimal measurement before and after a channel use is the same. In other words, a POVM $M=\{M_\x\}_{\x=1}^{n}$ that solves the state discrimination problem for the ensemble $S=\{q_\x ,\rho_\x\}_{\x=1}^{n}$, also solves it for $S^{(\N)}=\{q_\x ,\N[\rho_\x]\}_{\x=1}^{n}$. Moreover, $\pguess$ denotes the guessing probability of the original ensemble and $\pguess^{(\N)}$ the guessing probability after the channel use.
\end{defn}

We note that optimal measurement preservation is a property involving three objects: (i) an ensemble of states, (ii) a quantum channel, and (iii) an optimal measurement for discrimination.

\begin{defn}
	We call $\delta_\N= \pguess-\pguess^{(\N)}$ the \emph{guessing degradation}, which quantifies the decrease in the guessing probability after the use of a channel $\N$ for the ensemble $S$.
\end{defn}

 With these definitions we are ready to state the main result.
\begin{theorem*}
	\label{thereom OMP}
	Let $S=\{q_\x ,\rho_\x\}_{\x=1}^{n}$ denote an ensemble of $n$ qubit states, $\sigma_\x$ the complementary states of $S$ and $M=\{M_\x\}_{\x=1}^{n}$ an optimal measurement that identifies the states of the ensemble with indices from an index set $\A$. Then, a channel $\N$ with guessing degradation $\delta_\N$ is OMP for the ensemble $S$ if and only if the following conditions are satisfied:
	\begin{align}
	q_\x \N[\rho_\x] - q_\y \N[\rho_\y] & = q_\x \rho_\x - q_\y \rho_\y + \delta_\N(\sigma_\x -\sigma_\y) \,, \notag \\
	r_\x & \geq \delta_\N\geq 0  \, \,, \quad \forall \x,\y \in \A \,. \label{eq: theorem OMP iff}
	\end{align}
\end{theorem*}

\begin{proof}
	The sufficient part is straightforward. Assume that the conditions of the theorem are satisfied for the ensemble $S^{(\N)}$.  Then, they imply that
	\begin{align}
	q_\x \N[\rho_\x] - q_\y \N[\rho_\y] & = q_\x \rho_\x - q_\y \rho_\y + \delta_\N (\sigma_\x -\sigma_\y) \,, \notag \\
	& =  r_{\y} \sigma_{\y}  - r_{\x}\sigma_{\x}+ \delta_\N (\sigma_\x -\sigma_\y) \,, \notag \\
	& =  (r_{\y}-\delta_\N) \sigma_{\y}  - (r_{\x}-\delta_\N)\sigma_{\x} \,, \notag \\
	&\equiv  r_{\y}^{(\N)} \sigma_{\y} - r_{\x}^{(\N)}\sigma_{\x} \,. \label{eq: proof 1}
	\end{align}
	In the second equality we used the KKT conditions of the original ensemble, Eq.\@ \eqref{eq: 1st KKT geometric form}. 	The non-negativity of the $r_{\x}^{(\N)}$ is imposed by the second condition in the theorem.
	Eq.\@ \eqref{eq: proof 1} shows that the complementary states are the same for the two problems.
	Finally, it needs to be shown that the measurement is also the same for the two ensembles; this is, however, obvious since the complementary states are the same and thus the original measurement will satisfy the trace KKT condition: $\tr{\left[ M_\x \sigma_{\x}^{(\N)} \right]} = \tr{\left[ M_\x \, \sigma_{\x} \right]}  =0$ .
	
	It remains to establish the necessary part. Let us assume that the optimal measurement $M=\{M_\x\}_{\x=1}^n$ that identifies only states of the ensemble with indices $\A$, is preserved by the channel $\N$.
	 That means that if $M=\{M_\x\}_{\x=1}^n$ is a POVM that solves the discrimination problem for the original ensemble $S=\{q_\x ,\rho_\x\}_{\x=1}^{n}$, then the same POVM also solves the discrimination problem for the ensemble after the application of the channel $\N$, $S^{(\N)}=\{q_\x ,\N[\rho_\x]\}_{\x=1}^{n}$.
	 From qubit state discrimination  \cite{bae2013}, it is known that the POVM elements of an optimal measurement are necessarily weighted projectors, i.e. of the form $M_\x=w_\x \ketbra{\psi_\x^\perp}{\psi_\x^\perp}$, where $\sigma_\x=\ketbra{\psi_\x}{\psi_\x}$ are the complementary states of the original ensemble. Since the measurement operators have the aforementioned form, this immediately implies that the complementary states for the ensemble $S^{(\N)}$ are the same as the ones of the original one; that is, $\sigma^{(N)}_\x=\sigma_\x \,, \, \forall \x \in \A$. This follows from the second KKT condition:
	\begin{align}
	\tr{\left[ M_\x^{(\N)} \sigma_{\x}^{(\N)} \right]} = \tr{\left[ M_\x \, \sigma_{\x}^{(\N)} \right]}  =0 \,.
	\end{align}
	  It is worth noting that this last step holds true only in dimension two and fails in higher  dimensions; in the latter case it is possible for two ensembles to share an optimal measurement while their complementary states being different. Thus, the theorem becomes only sufficient if $\text{dim}>2$.

	To conclude the proof, let us write the KKT conditions subject to the constraint that the complementary states are the same for both problems and for any $\x\in \A$:
	\begin{align}
	q_\x \rho_\x - q_\y \rho_\y & = r_{\y} \sigma_{\y}  - r_{\x}\sigma_{\x} \,, \quad r_\x\geq 0 \notag \\
	q_\x \N[\rho_\x] - q_\y \N[\rho_\y] & =  r_{\y}^{(\N)} \sigma_{\y}  - r_{\x}^{(\N)}\sigma_{\x} \,,  \quad r_\x^{(\N)} \geq 0 \,. \label{eq: resulting KKT}
	\end{align}
	By subtracting the second equation from the first and noting that Eq\@ \eqref{eq: pguess KKT} implies
	\begin{align}
	r_\x-r_\x^{(\N)}=\pguess - \pguess^{(\N)} =\delta_\N \,, \label{eq: parameters to pguess} \quad \forall \x \,,
	\end{align}
	we obtain the condition in the theorem, Eq.\@ \eqref{eq: theorem OMP iff}. Moreover, from Eq.\@ \eqref{eq: parameters to pguess} and the constraint $r_\x^{(\N)} \geq 0$, we obtain the inequality constraint in the theorem, $r_\x \geq \delta_\N$.
\end{proof}

\vspace{1cm}
We note that the condition in Eq.\@ \eqref{eq: theorem OMP iff} can be concisely rewritten as  
\begin{equation}
\delta_\N (\sigma_\x - \sigma_\y)= (\N-\ido)(h_{\x\y}) \,,
\end{equation}
in terms of the \emph{Helstrom operator}, $h_{\x \y}\equiv q_\x \rho_\x - q_\y \rho_\y$, for each pair of states from the ensemble $S$ .

Let us compare the condition of the theorem with the previously derived OMP condition \cite{kechrimparis2019}:
\begin{align}
q_\x \N[\rho_\x] &- q_\y \N[\rho_\y]  = \kappa(q_\x \rho_\x - q_\y \rho_\y) \,, \notag \\
&  \quad \forall \x,\y \, \, \text{and} \, \,  \kappa  \in [0,1]  \,  \,. \label{eq: OMP old}
\end{align}
Note that by taking the trace on  both sides of last equation, the condition is self-consistent only if $\kappa=1$ for ensembles of unequal \emph{a priori} probabilities. This in turn implies that Eq.\@\eqref{eq: OMP old} can not be applied, in general, to ensembles of unequal probabilities. In addition, Eq.\@ \eqref{eq: OMP old} was only proven as a sufficient condition for a channel to be OMP. In contrast, both of these issues are addressed by the theorem, Eq.\@\eqref{eq: theorem OMP iff}. Indeed, taking the trace of both sides of Eq.\@\eqref{eq: theorem OMP iff} is always consistent, which extends the applicability of the previous OMP condition. Moreover, as it has been shown, the theorem is necessary and sufficient in dimension two, thus providing the full characterization of OMP channels. It remains valid as a sufficient condition in higher dimensions, with applicability to ensembles with unequal \emph{a priori} probabilities.

Let us introduce two notions of an OMP channel: (i) a \emph{strong} OMP channel that preserves all optimal measurements, and (ii) a \emph{weak} OMP channel that preserves only some of the optimal measurements or even just one of them. If a strong OMP channel exists, it preserves the full structure of the state discrimination problem, while a weak OMP channel only preserves the structure pertaining to a sub-ensemble of the original ensemble of states. In other words, in the former case the KKT conditions, Eq.\@ \eqref{eq: geometric KKT}, need to hold before and after the channel use for all values of the indices $\x$, while in the latter only for the values $\x \in \I$. Obviously, a weak OMP channel depends on the specific measurement(s) to be preserved; a strong OMP one is associated with the measurement that identifies all states in the ensemble, if it exists. 
It is clear from the above considerations that sets of weak OMP channels may overlap, as a channel may preserve a number of different measurements; they can also have only a trivial intersection, including the identity map only. 

As a concrete example consider the ensemble consisting of the four states in the Bennett-Brassard 1984 protocol \cite{bennett2014}, the eigenstates of the Pauli matrices $\hat{X}$ and $\hat{Z}$. Specifically, the states are $\left\{\ketbra{0}{0}, \ketbra{1}{1},\ketbra{+}{+}, \ketbra{-}{-}\right\}$, appearing with equal \emph{a priori} probabilities $\nicefrac{1}{4}$. Then, it is easy to see that one measurement that achieves the optimal guessing probability $\nicefrac{1}{2}$ is the one consisting of projectors to the states themselves, thus identifying all four states. Specifically, the POVM in this case is $M=\left\{\frac{\ketbra{0}{0}}{2}, \frac{\ketbra{1}{1}}{2},\frac{\ketbra{+}{+}}{2}, \frac{\ketbra{-}{-}}{2}\right\}$. However, $\hat{X}$ and $\hat{Z}$ measurements, with POVMs $M_{X}=\left\{\ketbra{+}{+}, \ketbra{-}{-}\right\}$ and $M_{Z}=\left\{\ketbra{0}{0}, \ketbra{1}{1}\right\}$, respectively, are also optimal. In the latter two cases, only two states are identified during a measurement. Consider a unitary channel that effects a rotation around the $z$ axis on the Bloch sphere. Obviously, only the $M_{Z}$ optimal measurement out of the three measurements is preserved. As a result, such a channel is weakly measurement preserving; in fact, for this ensemble there does not exist a unitary channel that is strongly measurement preserving. At the same time, as it will be shown later, the depolarization channel preserves all three measurements and thus belongs in the intersection of the three individual OMP sets; thus, it is an instance of an OMP channel belonging in the strong OMP set.


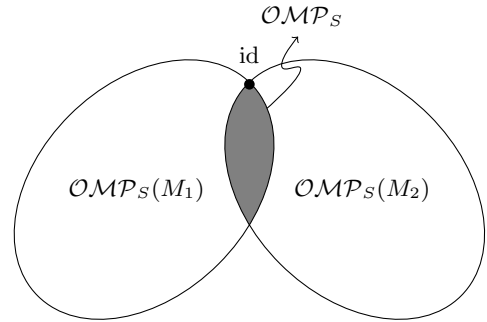
\begin{figure}[t]
	\begin{tikzpicture}[scale=2]
	
	
	
	\node at (0,0.9) {$\ido$};
	\node at (0.35,1.15) {$\OMP_S$};
	\node at (0.75,0) {$\OMP_{S}(M_2)$};
	\node at (-0.75,0) {$\OMP_{S}(M_1)$};

	\def\firstellipse{(0.7,0) ellipse [x radius=1, y radius =0.7, rotate=-45]}
	\def\secondellipse{(-0.7,0) ellipse [x radius=1, y radius =0.7, rotate=45]}
	\def\boundingbox{(-1.5,-0.5) rectangle (1.5,1.5)}
	
	
	\begin{scope}
	\clip \boundingbox;
	\clip \firstellipse;
	\fill[gray] \secondellipse;
	\end{scope}
	
	\draw \firstellipse \secondellipse;
	\begin{scope}[shift={(-0.82,-0.5)}]
	\draw [ ->] [black] plot [smooth, tension=1] coordinates {(48:1.4) (48:1.73) (52:1.68) (53:1.9)  };
	\end{scope}
	
	\fill[fill=black] (0,0.7) circle (1pt);

	\end{tikzpicture}
		\caption{Visualization of the OMP sets for an ensemble $S$. The left ellipse depicts the weak OMP set associated with a measurement $M_1$, while the right the weak OMP set associated with a measurement $M_2$. The intersection of the ellipses (in gray) corresponds to the intersection of the weak OMP sets which gives the strong OMP set, containing channels that preserves all optimal measurements.  The identity map is trivially contained in all OMP sets. } 
\end{figure}

Given an ensemble $S$ of qubit states, the theorem implies that the set of qubit channels is split into two subsets: channels that are OMP and those that are not. Note that by collecting \emph{all possible} OMP channels, weak and strong, into one set, they preserve different measurements, in general. Obviously, mixing channels that preserve different optimal measurement might not preserve any measurement at all and thus such a set in not convex. From a practical point of view, however, it is important to consider the OMP set for a \emph{given} optimal measurement of interest, e.g. a measurement that might have been prepared in some experiment. Thus, we separately highlight the set of strong OMP channels as well as different sets of weak OMP channels that preserve certain optimal measurements.   We denote the set of strong OMP channels of an ensemble $S$ with $\OMP_S$, while a set of weak OMP channels that preserves the optimal measurement $M$ with $\OMP_{S}(M)$. It is obvious that the strong OMP set is formed by taking the intersection of all the weak OMP sets. We have the following result:
\begin{prop}
	The set of OMP channels for an optimal measurement $M$ of an ensemble $S$ is convex.
\end{prop}
\begin{proof}
	This follows directly from the linearity of the conditions of the theorem and the convexity of the set of channels. Let $\N_1, \N_2 \in \OMP_{S}(M)$ (or $\OMP_{S}$) be two OMP channels of some ensemble $S$ that preserve an optimal measurement $M$, and $\N= (1-\kappa) \N_1 +\kappa \N_2$ a convex combination of the two. Writing the conditions of the theorem for each channel and taking their convex mixture shows that the channel $\N$ is also OMP, with guessing degradation the convex combination of the two individual guessing degradations of $\N_1$ and $\N_2$, i.e. $\delta_\N = (1-\kappa) \delta_{\N_1} +\kappa \delta_{\N_2}$ . 
\end{proof}



Let us examine the theorem in a number of special cases. First, we consider the case of equiprobable ensembles, i.e. ensembles of states that appear with equal \emph{a priori} probabilities. 
 The \emph{a priori} probabilities in this case are $q_\x=\nicefrac{1}{n} \,, \, \forall \x$, which implies that the conditions of the theorem become:
\begin{align}
\N[\rho_\x] - \N[\rho_\y] & =  \rho_\x -  \rho_\y + n \delta_\N (\sigma_\x -\sigma_\y) \,, \notag \\
1 & \geq n \delta_\N   \,. \label{eq: OMP iff equiprobable 1}
\end{align}
Noting that $\sigma_\x-\sigma_\y= \rho_\y-\rho_\x$ which follows from the KKT conditions for the original ensemble, Eq.\@ \eqref{eq: 1st KKT geometric form}, and substituting in last equation, we obtain
\begin{align}
\N[\rho_\x] - \N[\rho_\y] & =  \rho_\x -  \rho_\y + n \delta_\N \left(\rho_\y-\rho_\x \right) \notag \\
& =  (1-n \delta_\N) \left(\rho_\x -  \rho_\y \right)  \notag \\
& \equiv  \kappa\left(\rho_\x -  \rho_\y \right)  \,,
\end{align}
where we have defined the parameter $\kappa$,
\begin{align}
\kappa &= 1-n \delta_\N = 1- n\left(\pguess-\pguess^{(\N)} \right) \,.
\end{align}

Thus, the theorem implies that a channel $\N$ is OMP for an ensemble of $n$ equiprobable qubit states, $S^{(0)}=\{\frac{1}{n},\rho_\x\}_{\x=1}^{n}$, identified by a measurement $M=\{M_\x\}_{\x\in\A}$ if and only if
\begin{gather}
	\N[\rho_\x] - \N[\rho_\y] = \kappa\left(\rho_\x -  \rho_\y \right)  \\
	  \kappa \in (0,1] \,, \quad \x\,, \y \in \A \,.	\label{eq: OMP iff equiprobable}
\end{gather}
This is an agreement with the previous result in \cite{kechrimparis2019,kechrimparis2020,*kechrimparis2020a}. However, there results are now strengthened, since it is shown that the condition is necessary and sufficient in dimension two, while in \cite{kechrimparis2019,kechrimparis2020,*kechrimparis2020a} it was only shown to be sufficient. In addition, the structure of the parameter $\kappa$ is now understood to be a linear function of the guessing degradation $\delta_\N$.

Next we consider an ensemble that consists of a pair of states only, $S^{(2)}=\{q_\x,\rho_\x\}_{\x=1,2}$. In this case the problem is simplified since the two complementary states are two orthogonal projectors which obviously sum to the identity, $\sigma_1+\sigma_2=\id$, with their orthogonal complements being the two optimal POVM elements. That is, the optimal measurement is  $M=\{M_1,M_2\}=\{\ketbra{\psi}{\psi},\ketbra{\psi^\perp}{\psi^\perp}\}=\{\sigma_2,\sigma_1\}$, and in this case is unique. 
 If the Bloch vectors of the complementary states $\sigma_1,\sigma_2$ are denoted by $\vec{w}_1,\vec{w}_2$, they obey $\vec{w}_1=-\vec{w}_2\equiv \vec{p}$ and $\abs{\vec{w}_1},\abs{\vec{w}_2}=1$, from which it follows that $ \sigma_1-\sigma_2 =\vec{p}\cdot\vec{\sigma} \equiv \hat{P}$.
As a result, a measurement is OMP for an ensemble of two qubit states if and only if
	\begin{align}
	q_1 \N[\rho_1] - q_2 \N[\rho_2] & =  \left(q_1 \rho_1 - q_2 \rho_2\right) + \delta_\N  \hat{P} \notag \\
	r_i & \geq \delta_\N  \, \,, \quad i=1,2 \,, \label{eq: OMP iff two states}
	\end{align}
	where $\delta_\N$ is the guessing degradation and $\hat{P}=\vec{p}\cdot \vec{\sigma}$ is the observable that corresponds to the optimal measurement of the original ensemble.

An alternative formulation is the following. A measurement is OMP for an ensemble of two qubit states if and only if
\begin{align}
q_1 \N[\rho_1] - q_2 \N[\rho_2] & = \lambda \left(q_1 \rho_1 - q_2 \rho_2\right) + \mu \id \notag \\
r_i  \geq \delta_\N  \, \,, &\quad i=1,2 \,,  \label{eq: OMP iff two states 2}
\end{align}
where the parameters $\lambda$ and $\mu$ are explicitly given by
\begin{align}
\lambda &= \left(1-\frac{2\delta_\N}{r_1+r_2}\right) =  \frac{2\pguess^{(\N)}-1}{2\pguess-1} \,, \notag \\
\mu &=\delta_\N\frac{r_2-r_1}{r_1+r_2}= \delta_\N \frac{q_2-q_1}{2\pguess-1} \,.
\end{align}
The permissible values for the parameters $\lambda,\mu$ are 
\begin{align}
1\geq &\lambda \geq\frac{2q_1-1}{2\pguess-1}  \,, \notag \\
0 \geq &\mu \geq-\frac{q_1-q_2}{2}\frac{\pguess-q_1}{\pguess-\nicefrac{1}{2}} \,. \label{eq: lambda, mu range}
\end{align}

Note that the first condition in Eq.\@ \eqref{eq: OMP iff two states 2} can be concisely rewritten as 
\begin{equation}
\N(h_{12})=\lambda h_{12} +\mu \id \,,
\end{equation}
where $h_{12}$ denotes the Helstrom operator of the pair.
Thus, it follows that for a pair of states an optimal measurement is preserved by a quantum channel, if the resulting Helstrom operator is a certain linear combination of the original Helstrom operator and the identity.

\section{Characterization of qubit OMP channels for certain ensembles}
In this section, we derive general properties of OMP channels by restricting to certain types of ensembles of states. We start by recalling some known results on qubit channels.

A quantum channel is described by a \emph{completely positive, trace preserving} (CPTP) map. The set of qubit CPTP maps has been characterized in \cite{ruskai2002}. For a qubit state $\rho=\frac{1}{2} \left(\id + \vec{v}\cdot \vec{\sigma}\right)$, such a map can be written in the form:
\begin{equation}
\rho \rightarrow \N(\rho)=\frac{1}{2}\left(\id + (D\vec{v}+\vec{t})\cdot \vec{\sigma}\right)\,, \label{eq: qubit CPTP}
\end{equation}
where $D$ is some real $3\times 3$ matrix and $\vec{t}$ a vector with real entries. In geometric terms, the effect of a CPTP map is the transformation of the Bloch ball into a potentially displaced and rotated ellipsoid inside the ball; however, it's worth mentioning that not all ellipsoids in the interior correspond to legitimate CPTP maps \cite{ruskai2002}. The elements of $D$ and $\vec{t}$ are constrained in order for the map to be CPTP. Specifically, $D$ can be diagonalized with changes of bases to take the form $O_1 \Delta O_2$, where $O_1, O_2$ are rotations and $\Delta$ a diagonal matrix with real entries $\lambda_j$ that take values in $[-1,1]$. 

A unital map (i.e. $\vec{t}=0$) with diagonal $D$ is CPTP if and only if $(\lambda_1\pm\lambda_2)^2 \leq (1\pm \lambda_3)^2$. In general, a map as in Eq.\@ \eqref{eq: qubit CPTP} is CPTP if the map $\N_\Delta$ defined by
\begin{equation}
\N(\rho) = U \N_\Delta(V \rho \, V^\dagger) \, U^\dagger   \,.
\end{equation}
 for some unitaries $U,V$, is also CPTP. The conditions for $\N_\Delta$ to be CPTP are \cite{ruskai2002}:
 \begin{align}
(\lambda_1\pm\lambda_2)^2 \leq (1\pm \lambda_3)^2-t^2_3 \,, \label{eq: CPTP condition 1}
 \end{align}
 where $t_3$ denotes the third element of the vector $\vec{t}$ and
 \begin{align}
 &\left[1-(\lambda_1^2+\lambda_2^2+\lambda_3^2)-(t_1^2+t_2^2+t_3^2)\right]^2 \notag \\
 &\,\geq 4\left[\lambda_1^2(t_1^2+\lambda_2^2)+\lambda_2^2(t_2^2+\lambda_3^2)+\lambda_3^2(t_3^2+\lambda_1^2)-2\lambda_1\lambda_2\lambda_3\right] \,. \label{eq: CPTP condition 2}
 \end{align}
Note that if $\abs{\lambda_3}+\abs{t_3}\leq 1$ becomes an equality, then Eqs.\@ \eqref{eq: CPTP condition 1} and \eqref{eq: CPTP condition 2} are taken with $t_1=t_2=0$. 

Given a map of the form in Eq.\@ \eqref{eq: qubit CPTP} and using the conditions of the theorem, it follows that a map is OMP if and only if the following pairwise geometric conditions on the Bloch sphere are satisfied
\begin{align}
(D-\id)\left(q_\x \vec{v}_\x - q_\y \vec{v}_\y\right) + (q_\x -q_\y) \vec{t} -\delta_\N(\vec{s}_\x-\vec{s}_\y)=0 \,, \label{eq: Bloch OMP}
\end{align} 
$\forall \x,\y \in \A$ and $\delta_\N \leq r_\x$, and where $\vec{v}_\x$ denotes the Bloch vector of the state $\rho_\x$ and $\vec{s}_\x$ denotes the Bloch vector of the complementary state $\sigma_\x$. In other words, the channel $\N$ in Eq.\@ \eqref{eq: qubit CPTP}, is OMP if and only if the vector $\vec{t}$ can take the form
\begin{align}
\vec{t} = \frac{(\id-D)\left(q_\x \vec{v}_\x - q_\y \vec{v}_\y\right)}{q_\x-q_\y}  -\delta_\N \frac{(\vec{s}_\x-\vec{s}_\y)}{q_\x-q_\y} \,.
\end{align}
Note that if some of the \emph{a priori} probabilities are the same then for any such pair of indices, the constraint becomes a constraint on the matrix $D$ instead of the vector $\vec{t}$, which follows from Eq.\@ \eqref{eq: Bloch OMP}.

\subsection{Equiprobable ensembles}
Let us examine the set of OMP channels for equiprobable ensembles. For a channel of the form of Eq.\@ \eqref{eq: qubit CPTP}, the result in Eq.\@ \eqref{eq: OMP iff equiprobable} immediately implies the following geometric conditions $\forall \x,\y$:
\begin{align}
\left(D-\kappa \id\right) \left(\vec{v}_\x -\vec{v}_\y\right)\cdot \vec{\sigma}=0 \,,
\end{align}
which can be satisfied for all those maps for which all vectors $\vec{v}_\x-\vec{v}_\y$ are in the kernel of the matrix $D-\kappa \id$. Note that a map with $D=\kappa\id $ is always OMP, since it automatically satisfies the above conditions.
In addition, Eq.\@ \eqref{eq: OMP iff equiprobable} imposes the following restriction on the values of $\kappa$, namely $0<\kappa \leq 1$. Thus, we have obtained the following result.

\begin{prop}
	\label{prop: equiprobable}
	A channel is OMP for an equi-probable ensemble $S$ of $n$ qubit states, if and only if there exists a $\kappa\in(0,1]$ such that $(\vec{v}_\x-\vec{v}_\y )\in \ker{(D-\kappa \id)}$, where $\vec{v}_\x$ are the Bloch vectors of the states in the ensemble and $D$ is the real matrix in the definition of the map. 
	Moreover, a channel of the form
	\begin{equation}
	\Lambda (\rho) = \frac{1}{2}\left(\id + \left(\kappa \vec{v}+\vec{t}\, \right) \cdot \vec{\sigma} \right) \, \,, \quad  0<\kappa \leq 1 \,,
	\end{equation}
	is always OMP, where $\kappa$ and $\vec{t}$ are also constrained by the conditions for $\Lambda$ to be a CPTP map.
\end{prop}

It is straightforward to see that such a channel can also be written in the form
\begin{equation}
\Lambda(\rho) = (1-\eta) \rho +\eta \, \tau \,,
\end{equation}
with $1-\eta = \kappa$ and $\tau=\frac{1}{2}\left(\id+\frac{\vec{t}}{1-\kappa}\cdot \vec{\sigma}\right)$.  If  $\vec{t}=0$, then it follows that such a channel reduces to a depolarizing channel. Finally, if $\kappa=1$ the map corresponds to the identity. Note, however,  that these examples do not cover all cases allowed by Proposition \ref{prop: equiprobable} \,.

The results of this section show that the OMP set is never trivial for equi-probable ensembles.


\subsection{Two-state ensembles}

Having examined the case of equiprobable ensembles, we now consider ensembles of two states with strictly non-equal \emph{a priori} probabilities, ordered so that $q_1>q_2$. The case of equal \emph{a priori} probabilities is already covered by the results of last section. 

As it will be shown in a later section, for two states there always exist unitary transformations, $\E_U$ say, that preserve the optimal measurement. Moreover, from the convexity of the OMP set we know that any map of the form $ (1-\eta) \ido + \eta \E_U$ will also preserve the optimal measurement. As a result, the set of OMP maps for two-state ensembles is also nontrivial. 

Let us now derive the general form for a CPTP map to preserve the optimal measurement for an ensemble of two states.
Using Eq. \eqref{eq: qubit CPTP} and substituting in the conditions of the theorem for two states, Eq.\@ \eqref{eq: OMP iff two states 2}, we find the geometric conditions 
\begin{align}
\left(D-\lambda \id\right)  \vec{h} +  \vec{t}&=0 \,, \notag \\
(1-\lambda)(q_1-q_2)-2\mu &=0 \,,
\end{align}
where we have defined the vector
\begin{equation}
\vec{h} = \frac{\left(q_1 \vec{v}_1 -q_2 \vec{v}_2\right)}{(q_1-q_2)} \,,
\end{equation}
which is directly related to the Helstrom operator; it is in fact its Bloch vector, rescaled by the difference of \emph{a priori} probabilities. 
Thus, we have obtained the following result.
\begin{prop}
	Given an ensemble of two qubit states of non-equal \emph{a priori} probabilities and a CPTP map, the map is OMP if and only if there exist $ \mu, \lambda$ with values as in Eq.\@ \eqref{eq: lambda, mu range} such that
	\begin{align}
	\vec{t}=-(D-\lambda \id) \vec{h} \,, \\
	\mu = (1-\lambda) \frac{q_1-q_2}{2} \,.
	\end{align}
	Moreover, if the map is unital, i.e. $\vec{t}=0$, the first condition implies that a map is OMP if there exists a $ \lambda$ such that $\vec{h} \in \ker{(D-\lambda \id)}$\,.   
\end{prop}

\begin{figure}[]
	\includegraphics[]{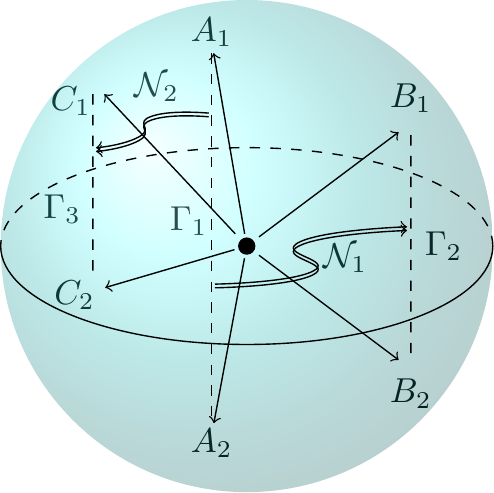}
	\caption{Visualization of a two-state ensemble $S=\{q_\x, \rho_\x\}_{\x=1,2}$ and two OMP channels $\N_1,\N_2$. $A_\x=q_\x \rho_\x$ denote the two states in the original ensemble, multiplied by the \emph{a priori} probabilities;  $B_\x=q_\x \N_1 (\rho_\x)$ and $C_\x=q_\x \N_2 (\rho_\x)$ represent the states multiplied by probabilities after the use of two OMP channels $\N_1, \N_2$ respectively. $\Gamma_1$ denotes the Helstrom operator of the original ensemble, while $\Gamma_2, \Gamma_3$ denote the ones of the ensembles after the use of channels $\N_1, \N_2$ respectively. The fact that the channels are OMP for the ensemble $S$ means that the line $\Gamma_1$ is parallel to $\Gamma_2, \Gamma_3$. Since the length of the lines $\Gamma_i$ is directly related to the guessing probabilities, the length of $\Gamma_1$ is always equal to or greater than those of $\Gamma_2, \Gamma_3$.}  
\end{figure}
It is instructive to recall the implications of Eq.\@ \eqref{eq: OMP iff two states}: a channel is OMP for a two-state ensemble if the Helstrom operator, $h^{(N)}_{12}$, of the resulting ensemble is equal to the the Helstrom operator, $h_{12}$, of the original ensemble plus the observable parallel to $h_{12}$, multiplied by some real number between zero and one. From this observation, it follows that a channel of the form
\begin{equation}
\E (\rho) = (1-\eta) \rho +\eta \frac{h_{12}}{q_1-q_2} \,,
\end{equation}
is OMP for any two qubit state ensemble. However, note that this is not the most general form of a channel allowed by Proposition 3. 

In general, the following geometric picture emerges. Any channel that transforms the two states in the ensemble in a way such that the Bloch vector of the Helstrom operator after the channel is parallel to the Bloch vector of the original one and with the same direction, then the channel is OMP. The length of the Bloch vector of the Helstrom operator does not have to be equal to the original in general, as long as it remains larger than the critical value $q_1-q_2$ \cite{hunter2003,weir2017}; otherwise, the measurement is not preserved and always guessing the most probable state without performing any measurement is the optimal strategy. Moreover, the length of the Bloch vector of the Helstrom operator is linearly related to the guessing probability and as a result can not increase after the use of the channel.

\section{Characterization of OMP properties of certain channels}

In this section we consider the reverse of the problem considered in last section: we fix a class of channels and look for ensembles for which they are OMP.


\subsection{Unitary channels}
Let us first examine the case of unitary channels. We have the following two propositions.
\begin{prop}\label{prop: unitary 2}
	A unitary map $\E_U(\rho)=U \rho \, U^\dagger $ is OMP  for an ensemble $S$ of two states, if and only if it leaves invariant the observable corresponding to the optimal measurement associated with the original Helstrom operator, $h_{12}=q_1 \rho_1-q_2 \rho_2$. 
\end{prop}
In other words, if a rotation by a unitary on the Bloch sphere is around the Bloch vector of the observable that corresponds to the optimal measurement for the original Helstrom operator, then the optimal measurement is preserved. 
The situation is different for ensembles of more than two states, as shown in next proposition.

\begin{prop}\label{prop: unitary 3}
	A unitary map $\E_U(\rho)=U \rho \, U^\dagger $ can not be OMP for an ensemble $S$ of $n>2$ states if all states are identified by the measurement. 
\end{prop}

\begin{proof}
	
	Both propositions follow from writing the first of the KKT conditions for the states in the ensemble $S$ and conjugating with the unitary $U$. Then, one finds that the parameters $r^{(\E_U)}_\x$ and the complementary states $\sigma^{(\E_U)}_\x$ after the channel use are given in terms of the original ones by 
	\begin{equation}
	r^{(\E_U)}_\x = r_\x  \, \,, \quad \sigma^{(\E_U)}_\x = U \sigma_\x \, U^\dagger \,,
	\end{equation} 
	which also shows that the guessing degradation is zero. Moreover, this shows that the measurement is not preserved, unless $\sigma_\x ^{(U)}=\sigma_\x$, which can only happen if complementary states are left invariant under the action of the unitary. This is only possible when the measurement consists of two projectors and the unitary effects a transformation along the axis parallel to the Bloch vector of the Helstrom operator.
\end{proof}

\subsection{Depolarization channel}
A depolarization channel is defined as
\begin{align}
\D[\rho]=(1-\eta) \rho +\eta \frac{\id}{2} \,, \quad (1-\eta) \in [0,1] \,.
\end{align}
We will show that the depolarization channel does satisfy the OMP conditions for equiprobable ensembles as well as a pair of states. We will also see why it fails for an ensemble of $n>2$ states with unequal \emph{a priori} probabilities.

Let us write the expression for the left hand side of Eq.\@ \eqref{eq: theorem OMP iff}, with the channel $\N$ being a depolarization channel:
\begin{align}
q_\x \D[\rho_\x]-q_\y \D[\rho_\y] = &(1-\eta)\left( q_\x \rho_\x -q_\y \rho_\y \right) \notag \\
&\quad + \eta \left(q_\x-q_\y\right) \frac{\id}{2} \,. \label{eq: LHS of OMP condition for depolarization}
\end{align}
It is obvious that the second term on the right hand side of last equation is the one that does not allow the depolarization channel to satisfy the OMP condition in general. However, the second term goes away for an ensemble of equal \emph{a priori} probabilities and the RHS reduces to Eq.\@ \eqref{eq: OMP iff equiprobable}, which also shows that it is OMP. Similarly, for an ensemble of two states only, the RHS has the form of Eq.\@ \eqref{eq: OMP iff two states 2}, which also shows that it is OMP directly. The fact that the depolarization channel is OMP for equiprobable ensembles was first shown in \cite{kechrimparis2019} by noting that it satisfies the condition in Eq.\@ \eqref{eq: OMP old}. For a two-state ensemble the situation is different, as it will not satisfy Eq.\@ \eqref{eq: OMP old} in general, which also confirms that it is only a sufficient condition. However, it was shown in \cite{kechrimparis2020,*kechrimparis2020a} that the depolarization channel is also OMP for two-state ensembles.
This observation was exploited  and a protocol was proposed to map any channel to an OMP one via an instance of a supermap, specifically channel twirling.

\section{Constructing an OMP set of an arbitrary qubit ensemble}

In this section, we derive the general solution to maps of the form of Eq.\@ \eqref{eq: qubit CPTP} that are consistent with the conditions of the theorem for a certain ensemble and for a given measurement to be preserved.

Let $S$ be the ensemble  in question and $M$ an optimal measurement that identifies the states with indices from an index set $\A$.
 The conditions of the theorem in the Bloch representation, Eq.\@ \eqref{eq: Bloch OMP}, become
\begin{align}
(D-\id)\vec{h}_{\x \y} + (q_\x -q_\y) \vec{t} -\delta_\N\vec{s}_{\x \y}=0 \,,\, \, \forall \, \x,\y \in \A\,, \label{eq: Bloch OMP simplified}
\end{align} 
where we defined $\vec{h}_{\x \y}=q_\x \vec{v}_\x - q_\y \vec{v}_\y=r_\y\vec{s}_\y-r_\x\vec{s}_\x$ and  $\vec{s}_{\x \y}= \vec{s}_\x-\vec{s}_\y $. 

Not all of these conditions are linearly independent since one can combine the conditions for pairs of indices to derive the condition for others. Specifically, denoting the left hand side of last equation with $R_{\x\y}$, then the condition is of the form $R_{\x\y}=0$. Moreover, it is trivial to notice that $R_{\x k} + R_{k \y}=R_{\x\y}$, which implies that if one has considered the condition for the index pairs $(\x,k)$ and $(k,\y)$, then the index pair $(\x,\y)$ has already been included. Let $m$ denote the number of elements in the index set $\A$. It is easy to see that the number of linear independent conditions are $m-1$. Let $a_j \in \A$; then the $m-1$ linearly independent conditions are explicitly
\begin{align}
(D-\id)\vec{h}_{a_1 a_2}& + (q_{a_1} -q_{a_2}) \vec{t} -\delta_\N\vec{s}_{a_1 a_2}&=0 \,, \notag \\
&\, \, \,  \vdots  \notag \\
(D-\id)\vec{h}_{a_1 a_{m-1}} &+ (q_{a_1} -q_{a_{m-1}}) \vec{t} -\delta_\N\vec{s}_{a_1 a_{m-1}}&=0 \,,  
\label{eq: Bloch OMP explicit}
\end{align}
where we arbitrarily chose index $a_1$ as the one appearing in each equation. 
Since we are interested in finding the channels that satisfy such conditions, the unknowns are the elements of $D$ and $\vec{t}$ in the definition of the map, as well as the guessing degradation $\delta_\N$. Note that after obtaining the general solution for $D, \vec{t}$ we are not done, since the conditions for the map to be CPTP, Eqs.\@ \eqref{eq: CPTP condition 1} and \eqref{eq: CPTP condition 2}, need to be imposed. As a result, we here obtain a set of \emph{feasible} solutions, which then need to be sieved by imposing Eq.\@ \eqref{eq: CPTP condition 1} and \eqref{eq: CPTP condition 2} in order to obtain the set of \emph{admissible} solutions. 

To turn Eqs.\@ \eqref{eq: Bloch OMP explicit} to a set of linear matrix equations in the standard form, we take the $i$-th component of each vector equation together to obtain the new set of equations
\begin{align}
H \vec{d}_1   +t_1 \vec{q} -\delta_\N \vec{w}_1 &= H \hat{e}_1 \,, \notag \\
H \vec{d}_2 + t_2 \vec{q} - \delta_\N \vec{w}_2 &= H \hat{e}_2 \,, \notag \\
H \vec{d}_3  + t_3 \vec{q} -\delta_\N \vec{w}_3 &= H \hat{e}_3 \,, \label{eq: linear Matrix OMP}
\end{align}
where we have defined the $(m-1) \times 3$ matrix
\begin{equation}
H= \begin{pmatrix}
\vec{h}_{a_1 a_{2}}^\top \\
\vdots \\
\vec{h}_{a_1 a_{m-1}}^\top \\
\end{pmatrix} \,,
\end{equation}
the vector of differences of \emph{a priori} probabilities
\begin{equation}
\vec{q}= \begin{pmatrix}
q_{a_1}-q_{a_2} \\
\vdots \\
q_{a_1}-q_{a_{m-1}} \\
\end{pmatrix} \,,
\end{equation}
and the vector of the $j$-th components of differences of the Bloch vectors of the complementary states 
\begin{equation}
\vec{w}_j= \begin{pmatrix}
(\vec{s}_{a_1 a_2})_j \\
\vdots \\
(\vec{s}_{a_1 a_{m-1}})_j  \\
\end{pmatrix} 
= \begin{pmatrix}
(\vec{s}_{a_1}- \vec{s}_{a_2})_j \\
\vdots \\
(\vec{s}_{a_1}- \vec{s}_{a_{m-1}})_j \\ 
\end{pmatrix} \,,
\end{equation}
with $j=1,2,3$, while $\hat{e}_j$ denote the vectors with the $j$-th element taking the value one, and zero elsewhere. The unknown vectors $\vec{d}_j$ are the rows of the matrix $D$ and $t_i$ are the elements of the vector $\vec{t}$  in the definition of a channel, Eq.\@ \eqref{eq: qubit CPTP}. That is,
\begin{equation}
D= \begin{pmatrix}
\vec{d}_{1}^{\top} \\
\vec{d}_{2}^{\top} \\
\vec{d}_{3}^{\top} \\
\end{pmatrix} \,.
\end{equation}

Let $\mathfrak{O}$ denote the $(m-1)\times 3$ matrix with zero as entries  and $\vec{0}$ the $(m-1)$-dimensional vector with zero entries. In addition, define the $3(m-1)\times 13$ matrix $Q$
\begin{equation}
Q= \begin{pmatrix}
H & \mathfrak{O}& \mathfrak{O} & \vec{q} & \vec{0} & \vec{0} & -\vec{w}_1 \\
\mathfrak{O}& H & \mathfrak{O}  & \vec{0} & \vec{q} & \vec{0} & -\vec{w}_2 \\
\mathfrak{O}  &  \mathfrak{O}& H &  \vec{0} & \vec{0} & \vec{q}  & -\vec{w}_3 \\
\end{pmatrix} \,,
\end{equation}
and the two vectors $\vec{x}$ and $\vec{b}$ given by
\begin{equation}
\vec{x}^\top= \begin{pmatrix}
\vec{d}_1^\top & \vec{d}_2^\top & \vec{d}_3^\top & \vec{t}^\top & \delta_\N \\
\end{pmatrix} \,,
\end{equation}
and
\begin{equation}
\vec{b}^\top= \begin{pmatrix}
\hat{e}_1^\top & \hat{e}_2^\top & \hat{e}_3^\top & \vec{0}^\top & 0 \\
\end{pmatrix}  \,,
\end{equation}
where $\vec{x}$ is the vector of unknowns. Then, the set of equations \eqref{eq: linear Matrix OMP}, take the standard form
\begin{equation}
Q\vec{x} = Q \vec{b} \,. \label{eq: inhomogenous}
\end{equation}
It is obvious that a particular solution is $\vec{x}_p=\vec{b}$, which also implies that the system is not inconsistent. The general solution is then
\begin{equation}
\vec{x} = \vec{x}_p + \vec{x}_h \,, 
\end{equation}
where $\vec{x}_h$ denotes the solution to the homogeneous equation
\begin{equation}
Q\vec{x} = 0 \,.
\end{equation}
The last matrix equation represents a linear set of $3(m-1)$ equations for 13 unknowns. Thus, depending on the number of states identified by the measurement, $m$, other solutions apart from the zero solution might not exist for the homogeneous problem. If $\rank(Q)<3(m-1)$, then an infinity of solutions exist; otherwise only one solution exists (identity map). 
Recalling that for ensembles of qubit states an optimal measurement identifying at most 4 states always exists \cite{davies1978}, it follows that even ensembles with more than 4 states will have some some of their weak OMP sets potentially non-trivial, that is, will have a non empty feasible set of solutions. 
An unresolved open question concerns the existence of ensembles with all their OMP sets trivial or, in other words, ensembles for which there does not exist other admissible solutions apart from the identity map.

Going back to the Eq.\@ \eqref{eq: inhomogenous}, the general solution can also be written concisely in the form \cite{james1978}
\begin{equation}
\vec{x} = \pinv{Q }Q \vec{b} +(\id -\pinv{Q} Q) \vec{c} \,, \label{eq: general solution OMP}
\end{equation}
for any $\vec{c}\in \mathbb{R}^{13}$ and where $\pinv{Q}$ denotes the pseudo-inverse of matrix $Q$. 
As already mentioned, this feasible set of solutions contains all possible maps of the form in Eq.\@ \eqref{eq: qubit CPTP} which are OMP; however not all of them are admissible since they don't have to be CPTP from the outset.
Having obtained the full set of solutions $\vec{x}$, it remains to impose the conditions for the map to be CPTP. This will further reduce the feasible set of solutions to the admissible set, which can in principle  be the trivial set.

\subsection{Examples}

In this section we derive the feasible solutions for a number of different ensembles. Note that we here consider only their strong OMP sets. At the same time, however, some of the examples are instances of weak OMP sets for others. For instance, the one basis or two mutually unbiased bases (MUB) examples are weak OMP sets for the three MUBs one.

\subsubsection{One basis}
Let us first consider a pair of orthogonal states. Specifically, we take the eigenstates of the Pauli $\hat{Z}$, which have Bloch vectors $\vec{v}^\top_{\nicefrac{1}{2}}  =\pm(0,0,1)$, and assume that they appear with probabilities $q_1$ and $q_2$ respectively. It is easy to see that in this case the complementary states have the same Bloch vectors as the states themselves but with signs inverted. That is, $\vec{s}_{\nicefrac{1}{2}}=-\vec{v}_{\nicefrac{1}{2}}$ In addition, the matrix $H$ becomes a row matrix and the vectors $\vec{q} \equiv q, \vec{w}_j\equiv w_j $ become scalars. Specifically, we find $H=(0 \, 0\, 1)\,, q=q_1-q_2 \,, w_1=w_2 =0$ and $w_3=-2$. From these, we form the matrix $Q$ and solve Eqs.\@ \eqref{eq: inhomogenous} to obtain 
\begin{equation}
D= \begin{pmatrix}
d_{11} & d_{12} & -(q_1-q_2)t_1  \\
d_{21} & d_{22} & -(q_1-q_2)t_2 \\
d_{31} & d_{32} & 1-2\delta_\N-(q_1-q_2)t_3   \\
\end{pmatrix} \,, 
\end{equation}
where $t_j$ denote the elements of the vector $\vec{t}$. Note that the guessing degradation $\delta_\N$ is here a free parameter and each allowed value specifies a different class of potentially OMP channels. Looking for unital OMP maps only, the first two elements of the third column of matrix $D$ become 0 while the last $1-2\delta_\N$.

\subsubsection{Two MUBs}
Next, we consider two MUBs. More specifically, we consider the ensemble of states used in the Bennett-Brassard 1984 cryptographic protocol \cite{bennett2014}. The four states 
have Bloch vectors $\vec{v}_{\nicefrac{1}{2}}  ^\top =\pm(0,0,1)$ and $\vec{v}_{\nicefrac{3}{4}}^\top  =\pm(1,0,0)$, and they appear with equal \emph{a priori} probabilities $\nicefrac{1}{4}$. A direct computation gives that the matrix $H$ is given by
\begin{equation}
H= \begin{pmatrix}
0 & 0 & \nicefrac{1}{2} \\
-\nicefrac{1}{4} & 0 & \nicefrac{1}{4} \\
\nicefrac{1}{4} & 0 &\nicefrac{1}{4} \\
\end{pmatrix} \,,
\end{equation}
while $\vec{w}_1^\top =2 (0,0,-1)$, $\vec{w}_2^\top = (1,0,-1)$ and $\vec{w}_3^\top =3 (-1,0,-1)$. Since the \emph{a priori} probabilities are equal, there are no constraints for the vector $\vec{t}$ in the definition of the channel. Thus, in this case
\begin{equation}
Q= \begin{pmatrix}
	H & \mathfrak{O}& \mathfrak{O}  & -\vec{w}_1 \\
	\mathfrak{O}& H & \mathfrak{O}  & -\vec{w}_2 \\
	\mathfrak{O}  &  \mathfrak{O}& H  & -\vec{w}_3 \\
\end{pmatrix} \,,
\end{equation}
while the two vectors $\vec{x}$ and $\vec{b}$ are
\begin{equation}
\vec{x}^\top= \begin{pmatrix}
\vec{d}_1^\top & \vec{d}_2^\top & \vec{d}_3^\top & \delta_\N \\
\end{pmatrix} \,,
\end{equation}
and
\begin{equation}
\vec{b}^\top= \begin{pmatrix}
\hat{e}_1^\top & \hat{e}_2^\top & \hat{e}_3^\top  & 0 \\
\end{pmatrix}                    \,,
\end{equation}
Finding the pseudo inverse of matrix $M$ and directly substituting in Eq.\@ \eqref{eq: general solution OMP} we find
\begin{equation}
\vec{x}^\top = \left(1-4\delta_\N,y,0,0,z,0,0,w, 1-4\delta_\N, \delta_\N\right) \,,
\end{equation}
from which the matrix $D$ of the OMP channel follows directly
\begin{equation}
D= \begin{pmatrix}
1-4\delta_\N & y & 0 \\
0 & z & 0\\
0 & w & 1-4\delta_\N \\
\end{pmatrix} \,. \label{eq: BB84 OMP}
\end{equation}
The parameter $\delta_\N$ is the guessing degradation and for the BB84 ensemble it takes values in $[0,\nicefrac{1}{4})$. Eq.\@ \eqref{eq: BB84 OMP} shows the general form of the matrix $D$ of an OMP map for the BB84 ensemble. However, not all such maps are CPTP. For example, consider a unital map with $y=z=w=1-4\delta_\N$; computing the singular value decomposition (SVD) of the matrix $D$ we find that it is of the form $O_1 \Delta O_2 $, with $\Delta = \text{diag}\left(1-4\delta_\N \,, \sqrt{2-\sqrt{3}}(1-4\delta_\N), \sqrt{2+\sqrt{3}}(1-4\delta_\N)\right)$ and $O_1, O_2$ rotations. Then, the condition $(\lambda_1\pm \lambda_2)^2 \leq (1\pm \lambda_3)^2$, gives the allowed values of $0 \leq \delta_\N \leq 0.078$ or $0.146 \leq \delta_\N \leq \nicefrac{1}{4}$, for the map to be CPTP. Similarly, let us consider the case of a non-unital channel with $y=z=w=1-4\delta_\N$, $t_1=t_3=0$ and $\delta_\N=\nicefrac{3}{10}$; then, the conditions for such a map to be CPTP, Eqs.\@ \eqref{eq: CPTP condition 1} and \eqref{eq: CPTP condition 2}, give $-0.678 \leq t_2 \leq 0.678$ .

\subsubsection{Three MUBs}
For completeness we also consider the case of three MUBs: the six states that appear with equal probability $\nicefrac{1}{6}$ appear in the \emph{six-state} cryptographic protocol \cite{bruss1998,bechmann1999}. 

The Bloch vectors of the six states are given by $\vec{v}_{\nicefrac{1}{2}}^\top =\pm(0,0,1) \,, \vec{v}_{\nicefrac{3}{4}}^\top  = \pm(1,0,0)\,, \vec{v}_{\nicefrac{5}{6}}^\top  = \pm(0,1,0)$. 
In this case, one finds that the matrix $H$ is given by
\begin{equation}
H= \frac{1}{6}\begin{pmatrix}
0 & 0 & 2 \\
-1 & 0 & 1 \\
1 & 0 & 1 \\
0 & -1 & 1 \\
0 & 1 & 1 \\
\end{pmatrix} \,,
\end{equation}
while $\vec{w}_1^\top =(0,1,-1,0,0)$, $\vec{w}_2^\top = (0,0,0,1,-1)$ and $\vec{w}_3^\top = -(2,1,1,1,1)$. After forming matrix $Q$ and solving Eq.\@ \eqref{eq: inhomogenous}, we find
\begin{equation}
\vec{x}^\top = \left((1-6\delta_\N),0,0,0,(1-6\delta_\N),0,0,0, (1-6\delta_\N), \delta_\N\right) \,,
\end{equation}
from which the matrix $D$ of the OMP channel follows directly
\begin{equation}
D= \begin{pmatrix}
1-6\delta_\N & 0 & 0 \\
0 & 1-6\delta_\N & 0\\
0 & 0 & 1-6\delta_\N \\
\end{pmatrix} \,, \label{eq: 3 MUBS}
\end{equation}
and with arbitrary $\vec{t}$. This shows that the only unital OMP maps in this case are depolarizing maps. 

\subsubsection{SIC-POVM state ensemble}
Let us now consider an equiprobable ensemble that consists of states that form a SIC-POVM \cite{renes2004}. The Bloch vectors of the 4 states are given by $\vec{v}_1^\top =(0,0,1) \,, \vec{v}_2^\top  = \left(\nicefrac{2\sqrt{2}}{3},0,-\nicefrac{1}{3}\right) \,, \vec{v}_3^\top  = \left(-\nicefrac{\sqrt{2}}{3},\sqrt{\nicefrac{2}{3}},-\nicefrac{1}{3}\right) \,, \vec{v}_4^\top  = \left(-\nicefrac{\sqrt{2}}{3},-\sqrt{\nicefrac{2}{3}},-\nicefrac{1}{3}\right)$. 
In this case, one finds that the matrix $H$ is given by
\begin{equation}
H= \frac{1}{4} \begin{pmatrix}
-\nicefrac{2\sqrt{2}}{3} & 0 & \nicefrac{4}{3} \\
\nicefrac{\sqrt{2}}{3} & -\nicefrac{\sqrt{2}}{3} & \nicefrac{4}{3} \\
\nicefrac{\sqrt{2}}{3} & \nicefrac{\sqrt{2}}{3} & \nicefrac{4}{3} \\
\end{pmatrix} \,,
\end{equation}
while $\vec{w}_1^\top =\nicefrac{(2\sqrt{2},-\sqrt{2},-\sqrt{2})}{3}$, $\vec{w}_2^\top = (0,\sqrt{\nicefrac{2}{3}},-\sqrt{\nicefrac{2}{3}})$ and $\vec{w}_3^\top = -\nicefrac{(4,4,4)}{3}$. Computing the matrix $Q$, its pseudo-inverse and substituting in Eq.\@ \eqref{eq: general solution OMP}, we find
\begin{equation}
\vec{x}^\top = \left((1-4\delta_\N),0,0,0,(1-4\delta_\N),0,0,0, (1-4\delta_\N), \delta_\N\right) \,,
\end{equation}
from which the matrix $D$ of the OMP channel follows directly
\begin{equation}
D= \begin{pmatrix}
1-4\delta_\N & 0 & 0 \\
0 & 1-4\delta_\N & 0\\
0 & 0 & 1-4\delta_\N \\
\end{pmatrix} \,, \label{eq: SIC-POVM OMP}
\end{equation}
and with arbitrary $\vec{t}$. This shows that the only unital OMP maps in this case are depolarizing maps. Note the similarity with the previous example.

\subsubsection{An ensemble of unequal \emph{a priori} probabilities}
We conclude the series of examples by considering an ensemble of unequal \emph{a priori} probabilities. Specifically, consider the three states with Bloch vectors $\vec{v}^\top_1=\nicefrac{(3,0,3)}{4\sqrt{2}} \,, \vec{v}^\top_2=\nicefrac{(-3,3\sqrt{3},0)}{10} \,, \vec{v}^\top_3=\nicefrac{(-1,-\sqrt{3},0)}{2}$ that appear with probabilities $q_1=\nicefrac{1}{3}, q_2=\nicefrac{5}{12}, q_3 = \nicefrac{1}{4}$, respectively. The $H$ matrix in this case is 2x3 and is given by
\begin{equation}
H= \begin{pmatrix}
\frac{1+\sqrt{2}}{8} & -\nicefrac{\sqrt{3}}{8} & \nicefrac{1}{4\sqrt{2}} \\
\frac{1+\sqrt{2}}{8} & \nicefrac{\sqrt{3}}{8} & \nicefrac{1}{4\sqrt{2}} \\
\end{pmatrix} \,.
\end{equation}
The complementary states are found using the results in \cite{ha2013} and have Bloch vectors $\vec{s}^\top_1=(-0.796, 0.385, -0.466)$, $\vec{s}^\top_2= (0.605, -0.713, 0.354)$ and $\vec{s}^\top_3= (0.304, 0.936, 0.178)$, from which one finds $\vec{w}^\top_1=(-1.401, -1.100), \vec{w}^\top_2=(1.098, -0.551)$ and $ \vec{w}^\top_3=(-0.821, -0.644)$. Although an exact calculation is possible, we only give the numerical values to avoid cumbersome expressions. Note that since the ensemble has only 3 states, the matrix $Q$ is now a $6\times13$ matrix. The full solution for the map to be OMP has the form
\begin{equation}
D= \begin{pmatrix}
d_{11} & d_{12} & 1.707-1.707 d_{11}-7.075 \delta_\N \\
d_{21} & d_{22} & -1.707 d_{21} +1.547 \delta_\N \\
d_{31} & d_{32} & 1-1.707d_{31}-4.145 \delta_\N  \\
\end{pmatrix} \,, 
\end{equation}
and $\vec{t} = (-2.598d_{12}+1.808\delta_\N, 2.598-2.598 d_{22}-9.894\delta_\N,-2.598d_{32}+1.059\delta_\N)$\,.
Once again, although such a map is OMP, it will not in general be CPTP and the conditions in Eqs.\@ \eqref{eq: CPTP condition 1} and \eqref{eq: CPTP condition 2} need to be imposed.

\section{Discussion}

\subsection{Guessing probability preservation}
The use of a channel can not increase the guessing probability and it will, at best, preserve it. Let us now examine what are the implications of the theorem in the case of no reduction in guessing probability, that is, $\delta_\N=0$. Then, the conditions of the theorem, Eq.\@ \eqref{eq: theorem OMP iff}, reduce to
\begin{align}
q_\x \N[\rho_\x] - q_\y \N[\rho_\y] & = q_\x \rho_\x - q_\y \rho_\y  \,, \notag \\
r_\x & \geq 0 \, \,, \quad \forall \x,\y \,. \label{eq: OMP iff no loss of guessing probability}
\end{align}
In other words, for the guessing probability to be preserved, pairwise differences shall be preserved for \emph{any} pair of states in the ensemble. In the qubit case, apart from the identity map, this is only possible for ensembles of two states and certain unitary maps, as shown in Propositions \ref{prop: unitary 2} and \ref{prop: unitary 3}. It follows that a unitary map preserves the guessing probability but not the optimal measurement, in general. Thus, it is obvious that for a general qubit ensemble, a channel cannot preserve both the guessing probability and the optimal measurement at the same time.

\section{Conclusions}

We proved a necessary and sufficient condition for the preservation of an optimal measurement for the discrimination of qubit states sent over a quantum channel. Our result contains and strengthens previous ones. We discussed particular simple forms of the condition in the case of ensembles of equal \emph{a priori} probabilities, as well as ensembles of two states. In addition, we considered the properties of the OMP sets for a given ensemble, which turn out to have a convex structure.
For ensembles of equal \emph{a priori} probabilities, as well as ensembles of two states, we further characterized OMP maps and showed that the depolarization channel is always included in their OMP sets. Thus, we showed that for these two cases the OMP set can not be the trivial set containing only the identity map. 
 Finally, we discussed how can one construct the OMP sets for a given ensemble and presented a number of examples.  
 
A few interesting open problems remain. 
 The first concerns the existence of ensembles whose OMP sets are all trivial.  This would imply that such ensembles are \emph{isolated} in that any CPTP map acting on them necessarily changes all optimal measurements. We demonstrated that two-state ensembles as well as ensembles of equal \emph{a priori} probabilities are not isolated. 
 Moreover, we showed the existence of feasible solutions for any ensemble.
 We expect that no isolated ensembles exist but we have not managed to establish the result in general. 
 
 The second concerns the extension of the result beyond the qubit case.
  In higher dimensions the theorem does not provide a necessary and sufficient condition since the link between the complementary states and the measurement breaks down: there might exist ensembles with the same measurement but with different complementary states. 
 Consequently, the condition in the theorem holds but as a sufficient condition only. 
 Owing to the intricacies of the state discrimination problem beyond qubit states, the existence of a simple condition that is both necessary and sufficient, similar to the one of the theorem in this work,  seems unlikely. However, by restricting to certain types of ensembles only, for example ensembles consisting of linearly independent states, a full characterization might be possible.

\section{Acknowledgment}
This work is supported by National Research Foundation of Korea (2019M3E4A1080001, NRF2017R1E1A1A03069961), an Institute of Information and Communications Technology Promotion (IITP) grant funded by the Korean government (MSIP) (Grant No. 2019-0-00831, EQGIS) and ITRC Program(IITP-2019-2018-0-01402).

\bibliography{OMPiff_20200630_sub.bib}
\bibliographystyle{apsrev4-1}

\end{document}